\documentclass[a4paper,11pt]{article}

\usepackage{graphicx} 
\usepackage{xspace}
\usepackage{comment}

\DeclareMathSymbol{\Gamma}{\mathalpha}{operators}{0}
\DeclareMathSymbol{\Delta}{\mathalpha}{operators}{1}
\DeclareMathSymbol{\Theta}{\mathalpha}{operators}{2}
\DeclareMathSymbol{\Lambda}{\mathalpha}{operators}{3}
\DeclareMathSymbol{\Xi}{\mathalpha}{operators}{4}
\DeclareMathSymbol{\Pi}{\mathalpha}{operators}{5}
\DeclareMathSymbol{\Sigma}{\mathalpha}{operators}{6}
\DeclareMathSymbol{\Upsilon}{\mathalpha}{operators}{7}
\DeclareMathSymbol{\Phi}{\mathalpha}{operators}{8}
\DeclareMathSymbol{\Psi}{\mathalpha}{operators}{9}
\DeclareMathSymbol{\Omega}{\mathalpha}{operators}{10}

\usepackage{amsmath,amssymb,amstext,amsthm,amsfonts}
\usepackage{graphicx}
\usepackage{tikz}

\newtheorem{theorem}{Theorem}[section]
\newtheorem{lemma}[theorem]{Lemma}

\newtheorem{definition}[theorem]{Definition}

\newtheorem{claim}{Claim}

\newcommand{\defparproblem}[4]{
  \vspace{1mm}
\noindent\fbox{
  \begin{minipage}{0.96\textwidth}
  \begin{tabular*}{\textwidth}{@{\extracolsep{\fill}}lr} #1  & {\bf{Parameter:}} #3 \\ \end{tabular*}
  {\bf{Input:}} #2  \\
  {\bf{Problem:}} #4
  \end{minipage}
  }
  \vspace{1mm}
}

\renewcommand{\H}{\mathcal{H}}
\newcommand{\cO}{O}

\newcommand{\eps}{\varepsilon}

\begin{document}
\title{Minimizing Rosenthal Potential in  Multicast Games\footnote{The research leading to these results has received funding from the European Research Council under the European Union's Seventh Framework Programme (FP/2007-2013) / ERC Grant Agreement n. 26795. 
This work is also supported by EPSRC (EP/G043434/1), Royal Society (JP100692), and Nederlandse Organisatie voor Wetenschappelijk Onderzoek (NWO), project: 'Space and Time Efficient Structural Improvements of Dynamic Programming Algorithms'.
A preliminary version of this paper appeared as an extended abstract in the proceedings
of ICALP 2012.
}
}

\author{Fedor V. Fomin\thanks{Department of Informatics, University of Bergen, PB 7803, 5020 Bergen, Norway. E-mails: {\tt{\{fedor.fomin,petr.golovach,michal.pilipczuk\}@ii.uib.no}}}
\addtocounter{footnote}{-1}
\and 
Petr A. Golovach\footnotemark
\and 
Jesper Nederlof\thanks{Utrecht University, Utrecht, the Netherlands. E-mail: \texttt{j.nederlof@uu.nl}}
\addtocounter{footnote}{-2}
\and 
Micha{\l}  Pilipczuk\footnotemark
\addtocounter{footnote}{1}
}

\date{}
\maketitle

\begin{abstract}
   A multicast game is a network design game modelling how selfish non-cooperative agents   build and maintain one-to-many network communication. There is a special source node and a collection of 
     agents located at corresponding terminals. 
   Each agent is  interested in selecting a route from the special source to its terminal minimizing the cost.
    The mutual influence of the agents is determined by a cost sharing mechanism, which  evenly splits the cost of an edge among all the agents using it for routing.
 In this paper we provide several algorithmic and complexity results on finding a Nash equilibrium minimizing the value of Rosenthal  potential. Let $n$ be the number of agents and $G$ be the communication network. We show that 
 \begin{itemize}
 \item  For a given  strategy profile $s$  and integer $k\geq 1$,
there is a local search algorithm which in time $n^{O(k)} \cdot  |G|^{O(1)}$  finds a better strategy profile, if there is any, in a $k$-exchange neighbourhood of $s$. In other words, the algorithm 
  decides  if  Rosenthal  potential  can be decreased  by changing strategies of 
 at most  $k$ agents;
 
 \item The running time of our local search  algorithm is essentially tight: unless  $FPT= W[1]$,   for any function $f(k)$, searching of the   $k$-neighbourhood cannot be done  in time $f(k)\cdot  |G|^{O(1)}$.
   \end{itemize}
The key ingredient of our algorithmic result is a subroutine that finds an equilibrium with minimum potential in $3^n \cdot  |G|^{O(1)}$ time. In other words, finding an equilibrium with minimum potential is fixed-parameter tractable when parameterized by the number of agents.
  \end{abstract}

\section{Introduction}\label{sec:intro}
Modern  networks are often  designed   and used by  non-cooperative individuals with diverse objectives. 
A considerable part of Algorithmic Game Theory focuses on optimization in such networks with selfish users \cite{Albers09,ChenRV10,EpsteinFM07,GuptaKPR07,GuptaST08,KoutsoupiasP09,RoughgardenS09,RoughgardenT02}.

In this paper we study the 
 conceptually simple but mathematically rich  cost-sharing model  introduced by Anshelevich et al. \cite{AnshelevichDKTWR08,AnshelevichDTW08}, see also  
\cite[Chapter~12]{KleinbergT05}. In a variant of the cost-sharing game, which was called by Chekuri et al. the \emph{multicast game}  \cite{ChekuriCLNO07}, the network is represented by a weighted directed graph with 
 a distinguished  source node $r$, and a collection of  $n$
     agents located at corresponding terminals.   
   Each agent is trying to select a cheapest  route from $r$  to its terminal.
   The mutual influence of the players is determined by a cost sharing mechanism identifying   how the cost of each edge in the network is shared among the agents using this edge.
When $h$ agents use an edge $e$ of cost $c_e$, each of them has to pay $c_e/h$. 
This is a very natural cost sharing formula which is also the outcome of the Shapley value.

The  multicast game belongs to the widely studied class of congestion games. This class of games was defined by
Rosenthal \cite{Rosenthal73}, who also 
 proved that every congestion game has a Nash equilibrium.   Rosenthal
 showed that for every congestion game it is possible to define a potential function 
which decreases if a player improves its selfish cost.  
Best-response dynamics in these games always lead to a set of paths that forms  a Nash equilibrium. Furthermore, every local minimum of 
Rosenthal potential  corresponds to a Nash equilibrium and vice versa.  However, while 
we know that the multicast game always has  a Nash equilibrium, the number of iterations in best-response dynamics achieving an equilibrium can be exponential (see \cite[Theorem 5.1]{AnshelevichDKTWR08}), and  it is an important open question if any Nash  equilibrium can be found in polynomial time.  
The next step in the study of Rosenthal potential  was done by  
Anshelevich et al. \cite{AnshelevichDKTWR08}, who  showed that Rosenthal potential  can be used not only for proving the  existence of a  Nash equilibrium  but also to estimate  the quality of equilibrium.  Anshelevich et al. defined the price of stability,  
  as   the ratio of the best Nash equilibrium  cost and the optimum network cost, the {\it social optimum}.  In particular, the cost of a Nash equilibrium 
  minimizing  Rosenthal potential is within $\log{n}$-factor of the social optimum, and thus the global minimum of the potential brings to a cheap equilibrium. The computational complexity of finding a Nash equilibrium achieving the bound of $\log{n}$ relative to the social optimum is still open, while computing the minimum of the Rosenthal potential  is NP-hard  \cite{AnshelevichDKTWR08,ChekuriCLNO07}.

 \paragraph{Our results.} In this paper we analyze the following  local search problem. Given a strategy profile $s$, we are interested if  a profile with a smaller value of  Rosenthal potential  can be found in a   $k$-exchange neighbourhood of $s$, which is the set of all profiles that can be obtained from $s$ by  changing  strategies of at most $k$ players.
 Our motivation to study this problem is two-fold.
 \begin{itemize}
  \item If we succeed in finding some Nash equilibrium, say by implementing best-response dynamics, which is still far from the social optimum, it is an important question if the already found equilibrium can be used to find a better one efficiently. Local search heuristic in this case is a natural approach.
  \item
 Since the number of iterations in best-response dynamics scenario can be exponential (see \cite[Theorem 5.1]{AnshelevichDKTWR08}), it can be useful to 
 combine the best-response dynamics with a heuristic that at some moments tries to make ``larger jumps", i.e., instead of decreasing Rosenthal potential  by changing strategy of one player, to decrease the potential by changing  in one step strategies of several players.  
 \end{itemize}
 Let us remark that the number of paths, and thus strategies, every player can select from, is exponential, so the 
 size of  the search space also can be exponential.  Since the size of $k$-exchange neighbourhood is exponential, 
 it is not clear a priori, if searching of a smaller value of Rosenthal potential in a  $k$-exchange neighbourhood of a given strategy profile can be done in polynomial time. We show that for a fixed $k$, the local search can be performed in polynomial time. The running time of our algorithm is $n^{O(k)} \cdot  |G|^{O(1)}$, where $n$ is the total number of players\footnote{The number of arithmetic operations used by our algorithms does not depend on the size of the input weights, i.e. the claimed running times are in the unit-cost model.}. As a subroutine, our algorithm uses a fixed-parameter algorithm  computing the minimum of  Rosenthal potential  in time $3^n \cdot  |G|^{O(1)}$. We find this auxiliary algorithm   to be interesting in its own. 
 It is known that for a number of  local search algorithms, exploration of the  $k$-exchange neighbourhood can be done by fixed-parameter tractable (in $k$) algorithms \cite{FellowsFLRSV12,Marx08a,Szeider09}.   We show that, unfortunately, this is not the case for the  local search of   Rosenthal potential  minimum.   
 We use tools from Parameterized Complexity, to show that the running time of our local search algorithm is essentially tight: unless  $FPT= W[1]$, searching of the $k$-neighbourhood cannot be done in time $f(k)\cdot  |G|^{O(1)}$ for any function $f(k)$.

\section{Preliminaries}
\paragraph{Multicast game and Rosenthal potential.}
A network is modeled by a directed $G = (V, E)$ graph.  
There is a special \emph{root} or \emph{source} node $r \in  V$. 
There are $n$ multicast users, \emph{players}, and each player has a specified \emph{terminal} node $t_i$ (several players can have the same terminals). 
A strategy $s^i$ for player $i$ is a path $P_i$ from $r$ to  $t_i$ in $G$. We denote by  $\Pi$ the set of players and by  
$S^i$ the finite set of strategies  of player $i$, which is the set of all paths from $r$ to $t_i$. The joint strategy space $S = S^1 \times S^2 \times \cdots \times S^n$ is the Cartesian product of all the possible strategy profiles.  At any given moment, a strategy
profile (or a configuration) of the game $s \in  S$  is the vector of all the strategies of the players, $s = (s^1, \dots , s^n)$.
Notice that for a given strategy profile $s$, several players can use paths that go through the same edge. For each edge $e\in E$ and a positive integer $h$, we have a cost $c_e(h)\in \mathbb{R}$ of the edge $e$ for each player who uses a path containing $e$, provided that exactly $h$ players share $e$.    
With each player $i$, we associate the cost function $c^i$ mapping a strategy profile $s\in S$ to real numbers, i.e., $c^i :S \to \mathbb{R}$.  For a strategy profile $s \in  S$, let $n_e(s)$ be the number of players using the edge $e$ in $s$. 
Then the cost the $i$-th player has to pay is 
 \[ c^i(s)=\sum_{e\in {E(P_i)}}c_e(n_e(s)),  
\]
and the total cost of $s$ is 
\[ c(s)=\sum_{i=1}^n c^i(s).  
\]
 The  \emph{potential} of   a strategy profile $s \in  S$, or equivalently, the set of paths $(P_1, \dots, P_n)$, is 
\begin{equation}\label{eq:potentialdef}
\Phi(s)=\sum_{e\in \cup_{i=1}^n E(P_i)}\sum_{h=1}^{n_e(s)}c_e(h).
\end{equation}

In this paper, we are especially interested in the case where the cost of every edge is split evenly between the players sharing it, i.e, the payment of player $i$ for edge $e$ is $c_e(h)=\frac{c_e}{h}$	for $c_e\in\mathbb{R}$. Respectively,   \emph{Rosenthal potential} of   a strategy profile $s \in  S$ is 
\[
\Phi(s)=\sum_{e\in \cup_{i=1}^n E(P_i)} c_e \cdot \H(n_e(s)),
\]
where $ \H(h) = 1 + 1/2 + 1/3 + \cdots  + 1/h$ is  the $h$-th Harmonic number.
 
For a strategy profile $s \in  S$ and $i \in \{1,2, \dots, n\}$, we denote by 
$s^{-i}$ the strategy profile of the players $j\neq i$, i.e. $s^{-i} = (s^1, \dots ,s^{i-1}, s^{i+1}, \dots  s^n)$. 
We use $(s^{-i}, \bar{s}^i)$ to denote the strategy profile identical to $s$, except that the $i$th player uses strategy $\bar{s}^i$ instead of ${s}^i$.
Similarly, for a subset of players $\Pi_0$, we define  $s^{-\Pi_0}$, the profile  of players $j\not\in\Pi_0$. For $\sigma\in \times_{i\in \Pi_0}S^i$, we denote by   $(s^{-\Pi_0}, {\sigma})$ the strategy profile obtained from  $s$ by changing the strategies of players in  $\Pi_0$ to $\sigma$.

 A strategy profile  $s \in S$ is a \emph{Nash equilibrium} if no player $i \in \Pi$ can benefit from unilaterally deviating from his action to another action, i.e., 
\[
\forall i \in  \Pi \text{ and } \forall \bar{s}^i \in S^i, 	 ~  c^i(s^{-i}, \bar{s}^i) \geq  c^i(s).
\]
The crucial property of Rosenthal potential $\Phi$ is that each step performed by a player improving his payoff also decreases $\Phi$.
Consequently, if $\Phi$ admits a minimal   value in strategy profile, this strategy profile is  a Nash equilibrium.

\paragraph{Parameterized complexity.}
We briefly review the relevant concepts of parameterized complexity theory that we employ.  For deeper background on the subject see the books by Downey and Fellows~\cite{DowneyF99}, Flum and Grohe~\cite{FlumGrohebook}, 
and Niedermeier~\cite{Niedermeierbook06}.

In the classical framework of P {\it vs} NP, there is only one measurement (the overall input size) that frames
the distinction between efficient and inefficient algorithms, and between tractable and intractable problems.
Parameterized complexity is essentially a two-dimensional sequel, where in addition to the overall input size $n$,
a secondary measurement $k$ (the {\it parameter}) is introduced, with the aim of capturing the contributions
to problem complexity due to such things as typical input structure, sizes of solutions, goodness of approximation, etc.
Here, the parameter is deployed as a measurement of the amount of current solution modification allowed in a local search
step.  The parameter can also represent an aggregrate of such bounds.

The central concept in parameterized complexity theory is the concept of {\it fixed-parameter tractability} (FPT), that is
solvability of the parameterized problem in time $f(k)\cdot n^{O(1)}$.  
The importance is that such a running time isolates all the exponential costs to a function of only the parameter.

The main hierarchy of parameterized complexity classes is
$$ FPT \subseteq W[1] \subseteq W[2] \subseteq \cdots \subseteq W[P] \subseteq XP.$$

The formal definition of classes $W[t]$ is technical, and, in fact, irrelevant to the scope of this paper. For our purposes it suffices to say that a problem is in a class if it is FPT-reducible to a complete problem in this class. 
Given two parameterized problems $\Pi$ and $\Pi'$, an \emph{FPT reduction} from $\Pi$ to $\Pi'$ maps an instance $(I,k)$
of $\Pi$ to an instance $(I',k')$ of $\Pi'$ such that 
\begin{itemize}
\item[(1)] $k' = h(k)$ for some computable function $h$, 
\item[(2)] $(I,k)$ is a {\sc YES}-instance of $\Pi$ if and only if $(I',k')$ is a {\sc YES}-instance of $\Pi'$, and 
\item[(3)] the mapping can be computed in FPT time.
\end{itemize} 

Hundreds of natural problems are known to be complete for the aforementioned classes, and $W[1]$ is considered the
parameterized analog of NP, because the {\sc $k$-Step Halting Problem} for nondeterministic Turing machines of
unlimited nondeterminism (trivially solvable by brute force in time $O(n^{k})$) is complete for $W[1]$. Thus, the statement $FPT\neq W[1]$ serves as a plausible complexity assumption for proving intractability results in parameterized complexity. {\sc Independent Set}, parameterized by solution size, is a more combinatorial example of a problem complete for $W[1]$. 
We refer the interested reader to the books by Downey and Fellows~\cite{DowneyF99} or Flum and Grohe~\cite{FlumGrohebook} for a more detailed introduction to the hierarchy of parameterized problems.

\paragraph{Local Search.} Local search algorithms are among the most common heuristics used to solve computationally hard optimization problems.
The common method of  local search algorithms is to move from solution to solution by applying local changes.
Books   \cite{EmileLbook,EmileMKbook} provide a nice introduction to the wide area of local search. 
 Best-response dynamics in congestion games corresponds to local search in $1$-exchange neighbourhood minimizing Rosenthal potential $\Phi$; improving moves for players decrease the value of the potential function. 
For  strategy profiles $s_1, s_2 \in  S$, we define the Hamming distance $D(s_1, s_2)=  |s_1\bigtriangleup s_2 |$  between $s_1$ and $s_2$, that is the number of  players implementing different strategies in $s_1$ and $s_2$.  We study the following parameterized version of the local search problem for multicast.

We define 
  \emph{arena} as  a directed graph $G$ with root vertex $r$, a multiset of target vertices $t_1,\ldots,t_\ell$ and for every edge $e$ of the graph a cost function $c_e: \mathbb{Z}^+ \rightarrow \mathbb{R}^+ \cup \{0\}$ such that 
 $c_e(h) \geq c(h+1)$ for $h \geq 1$.
 \medskip

\defparproblem{\textsc{$p$-Local Search  on  potential  $\Phi$}}{An arena consisting of graph $G$, vertices $r,(t_1,\ldots,t_\ell)$ and cost functions $c_e$, a strategy profile $s$, and an integer $k\geq 0$}{k}{Decide whether there is a strategy profile $s'$ such that $\Phi(s')<\Phi(s)$ and $D(s,s')\leq k$, where $\Phi$ is as defined in~\eqref{eq:potentialdef}.}

\section{Minimizing   Rosenthal potential}\label{sec:XP}
The aim of this section is to prove 
the following theorem.

\begin{theorem}\label{thm:RosenthalXP} The \textsc{$p$-Local Search  on  potential  $\Phi$} problem is solvable in time 
\[
\binom{|\Pi |}{k} \cdot 3^k \cdot  |G|^{\cO(1)}.
\]
\end{theorem}

Let us remark that in particular, if $\Phi$ is Rosenthal's potential, and hence the cost functions are of the special type $c_e(h)=\frac{c_e}{h}$, the \textsc{$p$-Local Search  on  potential  $\Phi$} problems can be solved within the running time of Theorem~\ref{thm:RosenthalXP}.

We need some additional terminology. Let $G$ be a directed graph.  
We say that a subdigraph $T$ of $G$ is an \emph{out-tree} if $T$ is a directed tree with only one vertex $r$ of in-degree zero (called the
\emph{root}). The vertices of $T$ of out-degree zero are called \emph{leaves}.
We also say that a strategy profile $s^*$ is \emph{optimal} if it gives the minimum value of the potential, i.e.,
for any other strategy profile $s$, $\Phi(s)\geq \Phi(s^*)$. 
Let $s=(P_1,\ldots,P_{|\Pi|})$ be a  strategy profile and $C\geq 1$ be an integer. We say that \emph{$s$ uses  $C$ arcs} 
if the union $T$ of the paths $P_i$ consists of   $C$ arcs.

If edge-sharing is profitable, then we can make the following observation about the structure of optimal strategies.

\begin{lemma}\label{lem:tree} Let $C$ be an integer such that there is a strategy profile using at most $C$ arcs. 
Let $s=(P_1,\ldots,P_{|\Pi|})$ be  a strategy profile using at most $C$ arcs such that
\begin{itemize}
\item[(i)]  Among all profiles using at most $C$ arcs, $s$ is   optimal. In other words,   for any profile $s'$ using at most $C$ arcs, we have 
 $\Phi(s')\geq \Phi(s)$.
 \item[(ii)] Subject to (i),  $S$ uses the minimum  number of arcs.
\end{itemize}
Then   the union $T$ of the paths $P_i$, $i\in \{1, \dots, |\Pi|\}$,  is an out-tree rooted in $r$.  
\end{lemma} 

\begin{proof}
Targeting towards  a contradiction, let us assume that
$T=\cup_{i=1}^{|\Pi|}P_i$ is not an out-tree.
Then there are paths $P_i,P_j$, $i,j\in\{1,\ldots,|\Pi|\}$, that have a common vertex $v\neq r$
such that the $(r,v)$-subpaths $P_i^v$ and $P_j^v$ of $P_i$ and $P_j$ respectively 
enter $v$ by different arcs. 

We show first  that 
\begin{equation}
\sum_{e\in E(P_i^v)}c_e(n_e(s))> \sum_{e\in E(P_j^v)}c_e(n_e(s)).
\label{ec:one}
\end{equation} 
cannot occur.  Assume that (\ref{ec:one}) holds. 
We claim that then the $i$-th player can improve his strategy and, consequently, $\Phi$ can be decreased, which will contradict the optimality of $s$. 
Denote by $P$ the $(r,t_i)$-walk obtained from $P_i$ by  replacing path $P_i^v$ by $P_j^v$. Notice that $P$ is not necessarily a path. Let $P'$ be a $(r,t_i)$-path in $P$
and let us construct the new strategy profile $s'=(s^{-i},P')$. This profile uses at most $C$ arcs. 
By non-negativity of $c_e(h)$, the new cost for the $i$-th player is equal to
\begin{align*}
\sum_{e\in E(P')}c_e(n_e(s'))&=&\sum_{e\in E(P')\cap E(P_i)}c_e(n_e(s))+\sum_{e\in E(P')\setminus E(P_i)}c_e(n_e(s)+1)&\leq\\
                          &\leq& \sum_{e\in E(P)\cap E(P_i)}c_e(n_e(s))+\sum_{e\in E(P)\setminus E(P_i)}c_e(n_e(s)+1).&\\
\end{align*}
Since for each $e\in E$ and $h\geq 1$, we have $c_e(h)\geq c_e(h+1)$, 
$$\sum_{e\in E(P)\setminus E(P_i)}c_e(n_e(s)+1)\leq\sum_{e\in E(P)\setminus E(P_i)}c_e(n_e(s)).$$
Therefore,
$$\sum_{e\in E(P')}c_e(n_e(s'))\leq\sum_{e\in E(P)}c_e(n_e(s)).$$
By (\ref{ec:one}), we have
$$\sum_{e\in E(P)}c_e(n_e(s))< \sum_{e\in E(P_i)}c_e(n_e(s)),$$ 
and the claim that player $i$ can improve follows.

Hence,
$$\sum_{e\in E(P_i^v)}c_e(n_e(s))\leq \sum_{e\in E(P_j^v)}c_e(n_e(s)).$$ 
By the same arguments as above, we can replace $P_j$ by a $(r,t_j)$-path $P$ in the walk 
obtained from $P_j$ by the replacement of $P_j^v$ by $P_i^v$ without increasing $\Phi$.
Moreover, we can repeat this operation for each path $P_h$, $h\neq i$, that 
enters $v$ by an arc that is different from the arc in $P_i$. The number of arcs used by the paths in the modified strategy profiles is at most the number of arcs used in $s$, and thus is at most $C$. 
It remains to observe that we obtain a strategy profile where $v$ has in-degree one in the union of paths. But it contradicts the choice of $s$, since we obtain a strategy profile that uses less arcs.  Hence, $T$ is an out-tree rooted in $r$.
\end{proof}

We use Lemma~\ref{lem:tree} to find an optimal strategy profile using the approach proposed by Dreyfus and Wagner~\cite{DreyfusW72} for the {\textsc{Steiner Tree}} problem. 

\begin{theorem}\label{thm:FPT} Given an arena as input, the minimum value of a potential $\Phi$ can be found in time 
$3^{|\Pi|} \cdot  |G|^{\cO(1)}$. The algorithm can also construct the corresponding optimal strategy profile $s^*$ within the same time complexity.
\end{theorem}

\begin{proof}
We give a dynamic programming algorithm. For simplicity, we only describe how to find the minimum of $\Phi$, but it is straightforward to modify the algorithm to obtain the corresponding strategy profile.

Let $T=\{t_1,\ldots,t_{|\Pi|}\}$ be the multiset of terminals. We construct partial solutions for subsets $X\subseteq T$. Also, while at the end we are interested in the answer for the source $r$, our partial solutions are constructed for all vertices of $G$. For a vertex $u\in V(G)$ and a multiset $X\subseteq T$, let $\Gamma_u^X$ denote the version of the game, in which only players associated with $X$ build paths from $u$ to their respective terminals. Therefore, we are interested in the game $\Gamma_r^T$. For a non-negative integer $m$, we define $\Psi(u,X,m)$ as the minimum value of the potential $\Phi(s)$ in the game $\Gamma_u^X$, taken over all strategy profiles $s$ such that the union of paths in $s$ contains at most $m$ arcs (we say that $s$ {\emph{uses}} arc $e$ if it is contained in some path from $s$). We assume that $\Psi(u,X,m)=+\infty$ if there are no feasible strategy profiles. 
Notice that by Lemma~\ref{lem:tree}, the number of arcs used in an optimal strategy in the original problem is at most $|V(G)|-1$.
Hence, our aim is to compute $\Psi(r,T,|V(G)|-1)$.

Clearly, $\Psi(u,\emptyset,m)=0$ for all $u\in V$ and $m\geq 0$. For non-empty $X$ and $m=0$, $\Psi(u,X,m)=0$ if all terminals in $X$ are equal to $u$, and $\Psi(u,X,m)=+\infty$ otherwise. We need the following claim.

\begin{claim}\label{cl:rec}
For $X\neq \emptyset$ and $m\geq 1$,  $\Psi(u,X,m)$ satisfies the following equation:
\begin{equation}
\begin{split}
\Psi(u,X,m)=\min\{&\Psi(u,X,m-1),\\
          &\Psi(u,X\setminus Y,m_1)+\Psi(v,Y,m_2)+\sum_{h=1}^{|Y|}c_{(u,v)}(h)\},        
\end{split}
\label{ec:claim}
\end{equation}
where the minimum is taken over all arcs $(u,v)\in E(G)$, $\emptyset \neq Y\subseteq X$, and $m_1,m_2\geq 0$ such that $m_1+m_2=m-1$; it is assumed that 
$\Psi(u,X,m)=\Psi(u,X,m-1)$ if the out-degree of $u$ is zero.     
\end{claim}

\begin{proof}
 Let 
$$\psi=\min\{\Psi(u,X,m-1),\Psi(u,X\setminus Y,m_1)+\Psi(v,Y,m_2)+\sum_{h=1}^{|Y|}c_{(u,v)}(h)\}.$$ 
We prove that $\Psi(u,X,m)=\psi$ by first showing that $\Psi(u,X,m)\geq\psi$, and then that $\Psi(u,X,m)\leq\psi$. Without loss of generality assume that $X=\{t_1,\ldots,t_{\ell}\}\subseteq T$, where $\ell=|X|$.

If $\Psi(u,X,m)=+\infty$, then $\Psi(u,X,m)\geq\psi$. Suppose that $\Psi(u,X,m)\neq+\infty$
and consider a strategy $s^*=(P_1,\ldots,P_{\ell})$ in the game $\Gamma_u^X$ which is optimal among those using at most $m$ arcs and, subject to this condition, the number of used arcs is minimum; in particular, $s^*$ has potential $\Psi(u,X,m)$. 
By Lemma~\ref{lem:tree}, $H=\cup_{i=1}^{\ell}P_i$ is an out-tree rooted in $u$. If $|E(H)|<m$, then
$\Psi(u,X,m)=\Psi(u,X,m-1)\geq \psi$. Assume that $|E(H)|=m$. As $m\geq 1$,  vertex $u$ has an out-neighbor $v$ in $H$. Denote by $H_1$ and $H_2$ the components of $H-(u,v)$, where $H_1$ is an out-tree rooted in $u$ and $H_2$ is an out-tree rooted in $v$. Let $Y\subseteq X$ be the multiset of terminals in $H_2$ and let $m_1=|E(H_1)|$, $m_2=|E(H_2)|$. Notice that exactly $|Y|$ players are using the arc $(u,v)$ in $s^*$ and $Y$ is nonempty.
Then $\Psi(u,X,m)\geq \Psi(u,X\setminus Y,m_1)+\Psi(v,Y,m_2)+\sum_{h=1}^{|Y|}c_{(u,v)}(h)\geq \psi$.

Now we prove that $\Psi(u,X,m)\leq\psi$. If $\psi=\Psi(u,X,m-1)$ then the claim is trivial, so let $v$, $Y$, $m_1$ and $m_2$ be such that $\psi=\Psi(u,X\setminus Y,m_1)+\Psi(v,Y,m_2)+\sum_{h=1}^{|Y|}c_{(u,v)}(h)$. 
Assume without loss of generality that $Y=\{t_1,\ldots,t_{\ell'}\}$ for some $\ell'\leq\ell$.
If $\Psi(u,X\setminus Y,m_1)=+\infty$ or $\Psi(v,Y,m_2)=+\infty$, then the inequality is trivial. Suppose that  $\Psi(u,X\setminus Y,m_1)\neq+\infty$ and $\Psi(v,Y,m_2)\neq+\infty$. Consider  
a strategy $s^*_1$ in the game $\Gamma_u^{X\setminus Y}$ that is optimal among those using at most $m_1$ arcs, and a strategy $s^*_2$ in the game $\Gamma_v^Y$ that is optimal among those using at most $m_2$ arcs. Of course, the potential of $s^*_1$ is equal to $\Psi(u,X\setminus Y,m_1)$, while the potential of $s^*_2$ is equal to $\Psi(u,Y,m_2)$.
We construct the strategy profile $s$ in the game $\Gamma_u^X$ as follows. For each terminal $t_j\in X\setminus Y$, the players use the $(u,t_j)$-path from $s_1^*$. For any $t_j\in Y$, the players use the $(v,t_j)$-path from $s_2^*$ after accessing $v$ from $u$ via the arc $(u,v)$, unless $u$ already lies on this $(v,t_j)$-path, in which case they simply use the corresponding subpath of the $(v,t_j)$-path. Note that $s$ uses at most $m_1+m_2+1=m$ arcs. Because for every $e\in E(G)$ and every $h\geq 1$,  we have that $c_e(h)\geq 0$, 
and $c_e(h)\geq c_e(h+1)$, 
we infer that $\Phi(s)\leq \psi$, as possible overlapping of arcs used in $s_1^*$, $s_2^*$ and the arc $(u,v)$ can only decrease the potential of $s$. Since $\Psi(u,X,m)\leq\Phi(s)$, this implies that $\Psi(u,X,m)\leq\psi$.
\end{proof}

In order to finish the proof of Theorem~\ref{thm:FPT}, we need to observe that using the recurrence (\ref{ec:claim}) one can compute the value $\Psi(r,T,|V(G)|-1)$ in time $3^{|\Pi|} \cdot  |G|^{\cO(1)}$.
\end{proof}

We use Theorem~\ref{thm:FPT} to construct  algorithm for \textsc{$p$-Local Search  on  potential  $\Phi$} and to conclude with the proof of Theorem~\ref{thm:RosenthalXP}.

\begin{proof}[Theorem of Theorem~\ref{thm:RosenthalXP}]
Consider an instance of \textsc{$p$-Local Search  on  potential  $\Phi$}. Let $T=\{t_1,\ldots,t_{|\Pi|}\}$ be the multiset of terminals and let $s$ be a strategy profile.
Recall that \textsc{$p$-Local Search  on  potential  $\Phi$} asks whether at most $k$ players can change their strategies in such a way that the potential decreases. Observe that we can assume that \emph{exactly} $k$ players are going to change their strategies because some of these players can choose their old strategies. There are $\binom{|\Pi|}{k}$ possibilities to choose a set of $k$ players $\Pi_0\subseteq \Pi$. We consider all possible choices and for each set $\Pi_0$, we check whether the players from this set can apply some strategy to decrease $\Phi$.   

Denote by $X\subseteq T$ the multiset of terminals of the players from $\Pi_0$, and let $s'=s^{-\Pi_0}$. We compute the potential $\Phi(s')$ for this strategy profile. Now we redefine the cost of edges as follows: for each $e\in E(G)$ and $h\geq 1$, $c'_e(h)=c_e(n_e(s')+h)$. 
Clearly, $c'_e(h)\geq 0$ and $c'_e(h)\geq c'_e(h+1)$. Let $\Phi'$ be the potential for these edge costs.
We find the minimum value of $\Phi'(s^*)$ for the set of players $\Pi_0$ and the corresponding terminals $X$. It remains to observe
that $\Phi(s')+\Phi'(s^*)=\min \{\Phi(s'')\ |\ s''=(s^{-\Pi_0},\sigma),\ \sigma\in \prod_{i\in \Pi_0}S^i \}$.
By Theorem~\ref{thm:FPT},  we can find $\Phi'(s^*)$ in time $3^k \cdot  |G|^{\cO(1)}$ and  the claim follows.
\end{proof}

\section{Intractability of local search for Rosenthal potential}\label{sect:hardness_rosenthal}

This section is devoted to the proof of 
the following theorem.

\begin{theorem}\label{thm:RosenthalW1-hard}  
\textsc{$p$-Local Search  on  potential  $\Phi$}, where  $\Phi$ is  Rosenthal  potential for multicasting game,  is W[1]-hard. 
\end{theorem}

Before we give the proof, let us remind a classical inequality that will be useful.

\begin{definition}
Let $a_1\geq a_2\geq \ldots \geq a_n$ and $b_1\geq b_2\geq \ldots \geq b_n$ be sequences of real numbers. We say that sequence $(a_i)$ {\emph{majorizes}} sequence $(b_i)$, denoted $(a_i)\succeq (b_i)$, if $\sum_{i=1}^n a_i=\sum_{i=1}^n b_i$ and $\sum_{i=1}^k a_i \geq \sum_{i=1}^k b_i$ for all $1\leq k<n$.
\end{definition}

\begin{theorem}[Hardy-Littlewood-Poly\'{a} inequality, \cite{hlp}]\label{hlp}
Let $f$ be a convex function on interval $[a,b]$ and $a_1\geq a_2\geq \ldots \geq a_n$ and $b_1\geq b_2\geq \ldots \geq b_n$ be sequences of real numbers from $[a,b]$. If $(a_i)\succeq (b_i)$, then 
$$\sum_{i=1}^n f(a_i)\geq \sum_{i=1}^n f(b_i).$$
\end{theorem}

Let us note that by changing the sign of $f$ we obtain that for concave functions the same result holds, but with inequality reversed.

We are now ready to prove Theorem~\ref{thm:RosenthalW1-hard}.

\begin{proof}
We provide an FPT reduction from the {\sc{Multicoloured Clique}} problem, which is known to be W[1]-hard~\cite{fellows-hermelin-rosamond-vialette-multicolored-hardness}.

  \medskip

\defparproblem{{\sc{Multicoloured Clique}}}{An undirected graph $H$ with vertices partitioned into $k$ sets $V_1,V_2,\ldots,V_k$, such that each set $V_i$ is an independent set in $H$.}{$k$}{Is there a clique $C$ in $G$ of size $k$?}

 \medskip

Observe that by the assumption that each $V_i$ is an independent set, the clique $C$ has to contain exactly one vertex from each part $V_i$.

We take an instance $(H,k)$ of {\sc{Multicoloured Clique}} and construct an instance $(G,s,k(k-1))$ of \textsc{$p$-Local Search on  potential  $\Phi$}. First, we provide the construction of the new instance; then, we prove that the constructed instance is equivalent to the input instance of {\sc{Multicoloured Clique}}. During the reduction we assume $k$ to be large enough; for constant $k$ we solve the instance $(H,k)$ in polynomial time by a brute-force search and output a trivial YES or NO instance of \textsc{$p$-Local Search   on  potential  $\Phi$}.

\medskip

\noindent{\bf{Construction.}} First create the root vertex $r$. For every $u\in V_i$, we create $k$ vertices: $\overline{u}$ and $u_1,\ldots,u_{i-1},u_{i+1},\ldots,u_k$. Denote by $F_u$ the set $\{u_1,\ldots,u_{i-1},u_{i+1},\ldots,u_k\}$. We connect the created vertices in the following manner: we construct one arc $(r,\overline{u})$ with cost $R=k^2$, and for all $j\in\{1,2,\ldots,i-1,i+1,\ldots,k\}$ we construct arc $(\overline{u},u_j)$ with cost $0$. With every vertex $u_j$ for all $u\in V(H)$ we associate a player that builds a path from $r$ to $u_j$. In the initial strategy profile $s$, each of $(k-1)|V(H)|$ players builds a path that leads to his vertex via the corresponding vertex $\overline{u}$. Observe that the potential of this strategy profile is equal to $|V(H)|\cdot R\cdot \H(k-1)$.

We now construct the part of the graph that is responsible for the choice of the clique. We create a {\emph{pseudo-root}} $r'$ and an arc $(r,r')$ with cost \[W=\frac{1}{\H(k(k-1))}\left(k\cdot R\cdot \H(k-1)-\frac{3}{2}\binom{k}{2}-\eps\right),\] where $\eps=\frac{k-1}{k^5}$. Note that $W\geq 1$ for sufficiently large $k$. For every edge $uv\in E(H)$, where $u\in V_i$ and $v\in V_j$, $i\neq j$,   we create a vertex $x_{uv}$,   arc $(r',x_{uv})$ of cost $1$, and arcs $(x_{uv},u_j)$, $(x_{uv},v_i)$ of cost $0$. This concludes the construction.

Before we proceed with the formal proof of the theorem, let us give some intuition behind the construction. 
Given a clique $C$ in $H$, we can construct a common strategy of $k(k-1)$ players assigned to vertices from $\bigcup_{u\in V(C)}F_u$, who can agree to jointly rebuild their paths via the pseudo-root $r'$. The ``cost of entrance"  for remodelling the strategy in this manner is paying for the expensive arc $(r,r')$; however, this can amortised by sharing cheap arcs $(r',x_{uv})$ for $uv\in E(C)$. The costs have been chosen so that only the maximum possibility of sharing, which corresponds to a clique in $H$, can yield a decrease of the potential.

\begin{figure}[htbp]
\begin{center}
\begin{tikzpicture}[scale=0.5]
  
 \tikzstyle{vertex}=[circle,fill=black,minimum size=0.09cm,inner sep=0pt]
 \node[vertex] (r) at (-12cm,0cm) {};
 \draw[left] (r) node {{\tiny{$r$}}};

 \foreach \x in {-7cm, -1cm, 9cm} 
 {
	\draw[rounded corners] (\x-0.9cm,-4cm) rectangle (\x+0.9cm,4cm);
 }

 \draw (-7.9cm,1.75cm) -- (-6.1cm,1.75cm);
 \draw (-7.9cm,1.25cm) -- (-6.1cm,1.25cm);
 \draw (-1.9cm,1.25cm) -- (-0.1cm,1.25cm);
 \draw (-1.9cm,0.75cm) -- (-0.1cm,0.75cm);

 \foreach \x/\z in {-0.75cm/1, -0.45cm/2, -0.15cm/3, 0.15cm/4, 0.45cm/5, 0.75cm/6} 
 {
	\node[vertex] (u1\z) at (\x-7cm,1.5cm) {};
	\node[vertex] (u2\z) at (\x-1cm,1cm) {};
 }

 \draw (-7cm,3cm) node {\tiny{\vdots}};
 \draw (-7cm,-1.25cm) node {\tiny{\vdots}};
 \draw (-1cm,2.75cm) node {\tiny{\vdots}};
 \draw (-1cm,-1.5cm) node {\tiny{\vdots}};
 \draw (4cm,0cm) node {\ldots};

 \draw (-7cm,-4.6cm) node {{\tiny{$V_1$}}};
 \draw (-1cm,-4.6cm) node {{\tiny{$V_2$}}};
 \draw (9cm,-4.6cm) node {{\tiny{$V_k$}}};

 \node[vertex] (ou) at (-8.5cm,0.5cm) {};
 \draw[below] (ou) node {{\tiny{$\overline{u}$}}};

 \draw[->] (r) -- (ou);
 \foreach \z in {1,2,3,4,5,6} 
 {
	\draw[->] (ou) -- (u1\z);
 }
 \draw[above] (-10.25cm, 0.25cm) node {{\tiny{$R$}}};

 \node[vertex] (rp) at (-10cm,6cm) {};
 \draw[above] (rp) node {{\tiny{$r'$}}};

 \draw[->] (r) -- (rp);
 \draw[left] (-11cm,3cm) node {{\tiny{$W$}}};

 \node[vertex] (x) at (-4cm,4cm) {};
 \draw[right] (x) node {{\tiny{$x_{uv}$}}};
 \draw[->] (rp) .. controls (-6cm,7cm) and (-5cm,7cm) .. (x);
 \draw (-6cm,6.8cm) node {{\tiny{$1$}}};

 \draw[->] (x) .. controls (-4cm,2cm) and (-7.45cm,3.5cm) .. (u12);
 \draw[->] (x) .. controls (-4cm,2cm) and (-1.75cm,3cm) .. (u21);

\end{tikzpicture}
\caption{Graph $G$}\label{fig:construction}
\end{center}
\end{figure}

\noindent{\bf{From a clique to a remodelled strategy profile.}} Assume that $C$ is a clique in $H$ with $k$ vertices. Let us remind, that in the initial strategy profile $s$ each  player is  using the corresponding arc  $(r,\overline{u})$ for his path. 
We construct the new strategy profile $s'$ by changing strategies of $k(k-1)$ players as follows. 
For every $uv\in E(C)$, where $u\in V_i$ and $v\in V_j$,  $i\neq j$, the players associated with vertices $u_j$ and $v_i$ reroute  their paths so that in $s'$ they lead via $r'$ and $x_{uv}$ to respective targets. In comparison to the profile $s$, the new profile $s'$:
\begin{itemize}
\item has congestion withdrawn from arcs $(r,\overline{u})$ for $u\in V(C)$---this decreases the potential by $k\cdot R\cdot \H(k-1)$;
\item has congestion introduced to arcs $(r,r')$ and $(r',x_{uv})$ for $uv\in E(C)$---this increases the potential by $W\cdot \H(k(k-1))+\frac{3}{2}\binom{k}{2}$.
\end{itemize}
Therefore, 
$\Phi(s')=\Phi(s)-k\cdot R\cdot \H(k-1)+W\cdot \H(k(k-1))+\frac{3}{2}\binom{k}{2}=\Phi(s)-\eps<\Phi(s).$

\noindent{\bf{From a remodelled strategy profile to a clique.}}
Let $s'$ be a strategy profile such that $\Phi(s')<\Phi(s)$ and $D(s,s')=\ell\leq k(k-1)$. Let $L$ be the set of players who have rebuilt their strategies in $s'$; then $|L|=\ell$. Let $p$ be a player in $L$, who is assigned to a vertex $u_j\in F_u$ for some $u\in V(H)$. Observe, that the only possibility of rebuilding the strategy for $p$ is to choose a path leading through $r'$ and a vertex $x_{uv}$ for some $uv\in E(H)$, $v\in V_j$. We now examine all the arcs of the graph $G$ with nonzero costs in order to provide a lower bound on $\Delta\Phi=\Phi(s')-\Phi(s)$. We partition the arcs into three classes: (i) arcs $(r,\overline{u})$ for $u\in V(H)$, (ii) arc $(r,r')$, and (iii) arcs $(r',x_{uv})$ for $uv\in E(H)$. For each of these classes we analyze the {\emph{contribution}} to the difference $\Delta \Phi$; by this we mean the difference of contributions to potentials $\Phi(s')$ and $\Phi(s)$ from the corresponding arcs.

Firstly, consider arcs $(r,\overline{u})$ for $u\in V(H)$. In total, $\ell$ players withdraw their paths from these arcs. The contribution to $s'$ of these arcs is equal to $\sum_{u\in V(H)} R\cdot \H(a_{\overline{u}})$, where $a_{\overline{u}}$ is the number of players using the arc $(r,\overline{u})$ in strategy profile $s'$. We know that $\sum_{u\in V(H)} a_{\overline{u}}=(k-1)|V(H)|-\ell$, while $0\leq a_{\overline{u}}\leq k-1$ for all $u\in V(H)$. Observe that then the sequence $(a_{\overline{u}})$ is majorized by a sequence consisting of $|V(H)|-\left\lfloor\frac{\ell}{k-1}\right\rfloor-1$ terms $(k-1)$, one term $(k-1) - (\ell\text{ mod }(k-1))$, and $\left\lfloor\frac{\ell}{k-1}\right\rfloor$ zeroes. Therefore, since $\H$ can be extended to a concave function, by Theorem~\ref{hlp} we infer that:
\begin{equation}\label{eq1}
\sum_{u\in V(H)} \H(a_{\overline{u}})\geq \left(|V(H)|-\left\lceil\frac{\ell}{k-1}\right\rceil\right)\H(k-1)+\H((k-1) - (\ell\text{ mod }(k-1))).
\end{equation}
Moreover, from concavity of function $\H$ we infer that 
\begin{equation}\label{eq2}
\H((k-1)-(\ell \text{ mod } (k-1)))\geq \frac{(k-1)-(\ell \text{ mod } (k-1))}{k-1}\H(k-1).
\end{equation}
Using (\ref{eq1}) and (\ref{eq2}) we infer that
\begin{equation}\label{eq3}
\sum_{u\in V(H)} \H(a_{\overline{u}})\geq \left(|V(G)|-\frac{\ell}{k-1}\right)\H(k-1).
\end{equation}
This implies that the contribution of these arcs to $\Delta\Phi$ is at least $-R\cdot\frac{\ell}{k-1}\cdot \H(k-1)$.

Now, consider the arc $(r,r')$. There are exactly $\ell$ players using this arc in $s'$, while in $s$ nobody was using it. Therefore, the contribution from this arc to $\Delta\Phi$ is equal to $W\cdot \H(\ell)$.

Finally, consider arcs $(r',x_{uv})$ for $uv\in E(H)$. All the $\ell$ players which rebuild their strategies in $s'$ use exactly one such arc. Moreover, each of these edges can be shared by at most two players. Therefore, the contribution to $\Delta\Phi$ from these edges is at least $\frac{\ell}{2}\cdot \H(2)=\frac{3}{4}\ell$, and the contribution is larger by at least $\frac{1}{2}$ if any player does not share the arc with some other player.

Concluding, since $\Phi(s')<\Phi(s)$ we have that $0>\Delta\Phi\geq \frac{3}{4}\ell+W\cdot \H(\ell)-R\cdot \frac{\ell}{k-1}\cdot \H(k-1)$. Therefore,
\begin{equation}\label{eq-main}
R\cdot\frac{\H(k-1)}{k-1}>\frac{3}{4}+W\cdot \frac{\H(\ell)}{\ell}.
\end{equation}
Note that once we infer that the contribution from any of the three classes of arcs is actually larger than estimeted in the previous paragraphs, we can add the corresponding term to the right-hand side of equation (\ref{eq-main}).
 
We claim now that $\ell=k(k-1)$. Let us define $g(t)=\frac{\H(t)}{t}$. Observe that for $t>1$ we have that
\begin{eqnarray*}
g(t)-g(t-1) & = &\frac{\H(t)}{t}-\frac{\H(t-1)}{t-1}=\H(t-1)\left(\frac{1}{t}-\frac{1}{t-1}\right)+\frac{1}{t^2} \\
& = & \frac{1}{t^2}-\frac{\H(t-1)}{t(t-1)}\leq \frac{1}{t^2}-\frac{1}{t(t-1)}\leq -\frac{1}{t^3}.
\end{eqnarray*}
Hence, function $g$ is decreasing and $\frac{\H(\ell)}{\ell}\geq \frac{\H(k(k-1))}{k(k-1)}$. 

Assume that $\ell<k(k-1)$; then it follows that $\frac{\H(\ell)}{\ell}\geq \frac{\H(k(k-1))}{k(k-1)}+\frac{1}{k^6}$. We obtain that
\begin{eqnarray*} 
R\cdot\frac{\H(k-1)}{k-1} & > & \frac{3}{4}+W\cdot \frac{\H(k(k-1))}{k(k-1)}+W\cdot \frac{1}{k^6} \\
& = & \frac{3}{4} + \frac{1}{k(k-1)}\left(k\cdot R\cdot \H(k-1)-\frac{3}{2}\binom{k}{2}-\eps\right)+W\cdot \frac{1}{k^6} \\
& = & \frac{3}{4} + R\cdot\frac{\H(k-1)}{k-1} -\frac{3}{4} -\frac{\eps}{k(k-1)} + W\cdot \frac{1}{k^6} \\
& = & R\cdot\frac{\H(k-1)}{k-1} + \frac{W-1}{k^6} \geq R\cdot\frac{\H(k-1)}{k-1}.
\end{eqnarray*}
The last inequality follows from $W\geq 1$, which is true for $k$ large enough. This contradiction shows that $\ell=k(k-1)$. \medskip

Every arc of the form $(r',x_{uv})$ can be used by at most two players. 
Now we want to prove that no arc $(r',x_{uv})$ can be used by exactly one player in the strategy profile $s'$.
For the sake of contradiction, we assume that at least one of the arcs $(r',x_{uv})$ is used by exactly one player in $s'$. Then the total contribution to $\Delta\Phi$ of the arcs of the form $(r',x_{uv})$ is at least $\frac{3}{4}\ell+\frac{1}{2}$. Similarly as before, we obtain that
\begin{eqnarray*} 
R\cdot\frac{\H(k-1)}{k-1} & > & \frac{3}{4}+\frac{1}{2k(k-1)}+W\cdot \frac{\H(k(k-1))}{k(k-1)} \\
& = & R\cdot\frac{\H(k-1)}{k-1} + \frac{1}{2k(k-1)} -\frac{\eps}{k(k-1)} \\
& = & R\cdot\frac{\H(k-1)}{k-1} + \left(\frac{1}{2}-\eps\right)\cdot\frac{1}{k(k-1)} \geq R\cdot\frac{\H(k-1)}{k-1}
\end{eqnarray*}
This contradiction shows that in the strategy profile $s'$, all the arcs of the form $(r',x_{uv})$ are used by zero or by two players.
\medskip

Finally, we want to prove that the sequence $(a_{\overline{u}})$ contains exactly $|V(H)|-k$ terms $k-1$ and $k$ zeroes, i.e. 
the set of players that did rebuild their strategies is ``concentrated" on $k$ vertices of $H$.
Suppose that this is not the case. Then the sequence $(a_{\overline{u}})$ is majorized by a sequence containing $|V(G)|-k-1$ terms $k-1$, one term $k-2$, one term $1$ and $k-1$ zeroes. By 
  Theorem~\ref{hlp},  the contribution of arcs of the form $(r,\overline{u})$ to $\Delta\Phi$ is at least 
\begin{eqnarray*}
 &&\hskip-1cm R  \cdot   (   ( |V(H)|- k-1)\cdot \H(k-1)    +    \H(k-2)+\H(1))  -R\cdot |V(H)|\cdot \H(k-1)\\ &= &
R\cdot (|V(H)|-k)\cdot \H(k-1)+R\cdot\left(1-\frac{1}{k-1}\right) -R\cdot |V(H)|\cdot \H(k-1)\\ &=& 
-R\cdot k\cdot \H(k-1) +R\cdot\left(1-\frac{1}{k-1}\right).
\end{eqnarray*}
Similarly as before, we obtain that
\begin{eqnarray*} 
R\cdot\frac{\H(k-1)}{k-1} & > & \frac{3}{4}+W\cdot \frac{\H(k(k-1))}{k(k-1)}+\frac{1}{k(k-1)}\cdot R\cdot\left(1-\frac{1}{k-1}\right) \\
& = & R\cdot\frac{\H(k-1)}{k-1} + \frac{R}{k(k-1)}\cdot\left(1-\frac{1}{k-1}\right) -\frac{\eps}{k(k-1)} \\
& \geq & R\cdot\frac{\H(k-1)}{k-1}.
\end{eqnarray*}
This contradiction shows that we can distinguish $k$ vertices $u^1,u^2,\ldots,u^k$ such that $L$ is exactly the set of players assigned to vertices $\bigcup_{i=1}^k F_{u^i}$. We claim that $H[\{u^1,u^2,\ldots,u^k\}]$ is a clique. 

Consider the vertex $u^1$. Without loss of generality we  assume that $u^1\in V_1$. In strategy profile $s'$, the player associated with vertex $u^1_j$ has to share an arc $(r',x_{u^1v})$ for some $v\in V_j$, for $j=2,3,\ldots,k$. Therefore, the set $\{u^2,\ldots,u^k\}$ has to contain a vertex from each of the sets $V_2,V_3,\ldots,V_k$; assume then, without loss of generality, that $u^j\in V_j$ for all $j=2,\ldots,k$.

Let us take $u^i$ and $u^j$ for $i\neq j$; we argue that $u^iu^j\in E(H)$, which will finish the proof. Consider players associated with vertices $u^i_j$ and $u^j_i$: they have to share an arc outgoing from $r'$, so there has to exist a vertex $x_{u^iu^j}$ and the corresponding arc $(r',x_{u^iu^j})$. From the construction of $G$ we infer that $u^iu^j\in E(H)$.
\end{proof}   

%{
%\bibliographystyle{siam}
%\bibliography{multicast}
%}

\end{document}